\documentclass[a4paper,english]{mylipics-v2016}
\usepackage{microtype}%if unwanted, comment out or use option "draft"

\title{Gossiping in Message-Passing Systems\footnote{Partly supported by ANR FREDDA and UMI RELAX.}}

\author{Benedikt Bollig}
\author{Marie Fortin}
\author{Paul Gastin}
\affil{LSV, CNRS \& ENS Paris-Saclay, Universit{\'e} Paris-Saclay\\
\texttt{\{bollig,fortin,gastin\}@lsv.fr}}

\authorrunning{B. Bollig, M. Fortin, and P. Gastin}

\Copyright{Benedikt Bollig, Marie Fortin, and Paul Gastin}

\makeatletter
\renewcommand{\paragraph}{\@startsection{paragraph}{6}{\z@}{2ex}{-0.7em}{\normalsize\bf}}
\makeatother

\usepackage{amsmath,amssymb}
\usepackage{stmaryrd}
\usepackage{fancyhdr}
\usepackage{xspace}
\usepackage{mathtools}
\usepackage{cite}
\usepackage{tikz}
\usetikzlibrary{trees,arrows,shapes,snakes,fit,shadows,decorations.text,decorations.pathmorphing,decorations.pathreplacing,calc,automata,positioning,patterns}

\usepackage{ifpdf}

\definecolor{mygreen}{RGB}{0.0,180,0.0}

\theoremstyle{plain}
\newtheorem{proposition}[theorem]{Proposition}

\usepackage{amsmath,amssymb}
\usepackage{stmaryrd}
\usepackage{xspace}
\usepackage{mathtools}
\usepackage{tikz}
\usetikzlibrary{trees,arrows,shapes,snakes,fit,shadows,decorations.text,decorations.pathmorphing,decorations.pathreplacing,calc,automata,positioning,patterns,arrows.meta}

\tikzstyle{dot} = [circle, fill, inner sep=0, minimum size = 2pt]
\tikzstyle{bdot} = [circle, fill, inner sep=0, minimum size = 4pt]

\definecolor{mygreen}{RGB}{0.0,180,0.0}

\newcommand\A{\ensuremath{\mathcal{A}}\xspace}

\newcommand\prel{\rightarrow}
\newcommand\mrel{\lhd}
\newcommand\init{\iota}

\newcommand\CFM{\text{CFM}\xspace}
\newcommand\CFMs{\text{CFMs}\xspace}

\newcommand\Acc{\mathit{Acc}}

\newcommand\rPastp[2]{{\downarrow_{#1}} (#2)}

\newcommand\p{\pi}

\newcommand\source{\mathit{source}}
\newcommand\target{\mathit{target}}
\newcommand\tmsg{\mathit{msg}}
\newcommand\tlabel{\mathit{label}}
\newcommand\receiver{\mathit{receiver}}
\newcommand\sender{\mathit{sender}}
\newcommand\MSCs[2]{\mathbb{MSC}(#1,#2)}

\newcommand\Paths{\Pi}
\newcommand\last[2]{\mathsf{pred}_{#1}(#2)}
\newcommand\first[2]{\mathsf{succ}_{#1}(#2)}
\newcommand\lastf[1]{\mathsf{pred}_{#1}}
\newcommand\firstf[1]{\mathsf{succ}_{#1}}

\newcommand\Procs{P}

\newcommand\Msg{\mathit{Msg}}
\newcommand\msg{m}

\newcommand\lastp[2]{\mathsf{latest}_{#1}(#2)}
\newcommand\fa[2]{\alpha_{#1,#2}}

\newcommand\Ord{\mathcal{R}}

\newcommand\Afa[2]{\A_{\alpha,#1,#2}}

\newcommand\Ale{\A_\ple}

\newcommand\botp[1]{\bot}
\newcommand\topp[1]{\top}

\newcommand\X[1]{\mathop{\mathsf{X}_{#1}}}
\newcommand\Y[1]{\mathop{\mathsf{Y}_{#1}}}
\newcommand\Co{\mathop{\mathsf{Co}}}
\newcommand\Up[1]{\mathop{\mathcal{U}_{#1}}}
\newcommand\Op[1]{\mathop{\mathcal{O}_{#1}}}
\newcommand{\Until}{\mathbin{\tilde{\mathsf{U}}}}
\newcommand{\Since}{\mathbin{\tilde{\mathsf{S}}}}
\newcommand\TL{\ensuremath{\textup{LTL}(\mathsf{Co},\Until,\Since)}\xspace}

\newcommand{\moves}{\Gamma}

\newcommand{\pmove}{{\prel}}
\newcommand{\pamove}{\mathrel{\text{$\xrightarrow{\ast}$}}}
\newcommand{\ppmove}{\mathrel{\text{$\xrightarrow{+}$}}}
\newcommand{\mmove}[2]{{\mrel_{#1,#2}}}
\newcommand{\amove}[1]{#1}
\newcommand{\sem}[2]{\llbracket {#2} \rrbracket_{#1}}
\newcommand{\Pto}[4]{(#2,#4) \in \sem {#1} {#3}}

\newcommand{\Comp}[1]{\mathit{Comp}(#1)}
\newcommand{\loc}[1]{\langle#1\rangle}
\newcommand{\send}[3]{\langle#1,!_{#3}#2\rangle}
\newcommand{\rec}[3]{\langle#1,?_{#3}#2\rangle}
\newcommand{\C}{\mathcal{C}}
\newcommand{\df}{:=}
\newcommand{\ple}{\preceq}
\newcommand{\pless}{\prec}
\newcommand{\act}{\gamma}
\newcommand{\ploc}{\mathit{loc}}
\newcommand{\lab}{\lambda}
\newcommand{\nbullet}{\resizebox{!}{1.2ex}{$\diamond$}}

\newcommand{\abullet}{
\begin{tikzpicture}
\filldraw[fill=orange!60,draw=black] circle (2.8pt);
\end{tikzpicture}
}

\newcommand{\bbullet}{
\begin{tikzpicture}
\filldraw[fill=blue!50,draw=black] rectangle (5pt,5pt);
\end{tikzpicture}
}

\newcommand{\Agossip}{\mathcal{A}_{\mathsf{gossip}}}
\newcommand{\pexpr}[1]{\p_{#1}}
\newcommand{\SimplePaths}{\Pi^{\mathsf{gossip}}}

\newcommand{\colone}{
\begin{tikzpicture}
\filldraw[fill=white,draw=black] circle (2.8pt);
\end{tikzpicture}
}
\newcommand{\coltwo}{
\begin{tikzpicture}
\filldraw[fill=gray!90,draw=black] circle (2.8pt);
\end{tikzpicture}
}
\newcommand{\zcolone}{
\begin{tikzpicture}
\filldraw[fill=white,draw=black] rectangle (5pt,5pt);
\end{tikzpicture}
}
\newcommand{\zcoltwo}{
\begin{tikzpicture}
\filldraw[fill=gray!90,draw=black] rectangle (5pt,5pt);
\end{tikzpicture}
}

\newcommand{\xia}[1]{\xi_1(#1)}
\newcommand{\xib}[1]{\xi_2(#1)}
\newcommand{\tht}[2]{\theta(#1)(#2)}

\begin{document}

\maketitle

\pagestyle{fancy}
\fancyhead{}
\renewcommand{\headrulewidth}{0pt}
\fancyfoot[C]{\vspace{2ex}\thepage}

% ------------------------------------------------------------
% ------------------------------------------------------------
% ------------------------------------------------------------

\begin{abstract}
We study the gossip problem in a message-passing environment:
When a process receives a message, it has to decide whether
the sender has more recent information on other processes than itself.
This problem is at the heart of many distributed algorithms,
and it is tightly related to questions from formal methods
concerning the expressive power of distributed automata.
We provide a non-deterministic gossip protocol for message-passing
systems with unbounded FIFO channels,
using only finitely many local states and a finite message alphabet.
We show that this is optimal in the sense that there is no deterministic
counterpart. As an application, the gossip protocol allows us to show
that message-passing systems capture well-known extensions of
linear-time temporal logics to a concurrent setting.
\end{abstract}

\section{Introduction}

\emph{Causality} is a fundamental concept in distributed computing \cite{Attiya:2004,Raynal:2013,Lynch:1996,Tel:2001}. In his influential paper \cite{Lamport78}, Lamport postulated that events in an execution of a distributed system are partially ordered by what is commonly referred to as the happens-before or causal-precedence relation. Two events that are related in the partial order can be considered \emph{causally dependent}. Tightly related is the notion of a snapshot, or global system state, which corresponds to a ``lateral cut'' through the partial order. Snapshot computations are at the heart of many distributed algorithms such as deadlock and termination detection, checkpointing, or monitoring. However, they are intricate due to the absence of a shared memory and unpredictable delay of message delivery, and they continue to constitute a fundamental research area \cite{Raynal:2013}.

\smallskip

A variety of techniques exist to obtain a consistent view of the global system state, ranging from time-stamping to ``gossiping''.
The aim of the latter is to keep track of the latest information that a process has about all other processes.
Interestingly, gossip protocols and related techniques such as asynchronous mappings have also been exploited in formal methods, in particular when it comes to establishing the expressive power of an automata model \cite{MukundS97,CMZ93,MukundKS03,DolevS97}. In particular, gossip protocols are the key to simulating high-level specifications, which include message sequence graphs and monadic second-order logic \cite{HenriksenJournal,GKM06,Kuske01,Zielonka87,tho90traces}. All these techniques and algorithms, however, require that communication be synchronous or accomplished through FIFO channels with \emph{limited} capacity.

Now, it is a standard assumption in distributed computing that channels are a priori unbounded (cf.\ \cite{Raynal:2013,Tel:2001}). In this paper, we consider the gossip problem in a message-passing environment where a finite number of processes communicate through \emph{unbounded} point-to-point FIFO channels. The problem can be stated as follows:
\begin{center}
  \parbox{0.6\textwidth}{ \it Whenever process $q$ receives a message from process $r$, $q$ has to decide, for all processes $p$, whether it has more recent information on $p$ than $r$.  }
\end{center}
Equivalently, $q$ has to output the most recent local state of $p$ that is still in its causal past.
The gossip protocol is superimposed on an existing system.
It is \emph{passive} (also reactive or observational) in the sense that it can add information to messages that \emph{are sent anyway}.
It is neither allowed to initiate extra communications nor to suspend the system activity.
This is fundamentally different from classical snapshot algorithms such as the one by Chandy and Lamport \cite{ChandyL85}, where the system is allowed to intersperse new send and receive events. In fact, like \cite{MukundS97,CMZ93,MukundKS03}, we will impose additional requirements: Both the set of messages and the set of local states must be finite. Besides being a natural assumption, this will allow us to exploit the gossip protocol to compare the expressive power of temporal logics and message-passing systems.

However, we will show that, unfortunately, there is no \emph{deterministic} gossip protocol.
This impossibility result is in contrast to the deterministic protocols for synchronous communication or
message-passing environments with bounded channels \cite{MukundS97,CMZ93,MukundKS03,DolevS97}.

On the positive side, and as our main contribution, we provide a non-deterministic gossip protocol: For every possible communication scenario,
\begin{itemize}
\item there is an accepting run that produces the correct output (i.e., the correct latest information);
\item there may be system runs that do not produce the correct output, but 
these runs will be rejected by our gossip protocol.
\end{itemize}

The (non-deterministic) gossip protocol is an important step towards a better understanding of the expressive power of communicating finite-state machines (CFMs), which are a classical model of message-passing systems \cite{Brand1983}.
From a logical point of view, maintaining the latest information in a distributed system is a first-order property that requires \emph{three} variables: An event $e$ on process $p$ is the most recent one in the causal past of an event $f$ if all other events $g$ on $p$ that are in the causal past of $f$ are also in the past of~$e$.
Unfortunately, it is not known whether first-order formulas can always be translated into communicating finite-state machines.
However, using our gossip protocol, we show that we can deal with all formulas from classical temporal logics that have been studied for concurrent systems in the realm of partial orders \cite{Thiagarajan94,GK-fi07,DiekertG06}.
Since gossiping has been employed for implementing other high-level specifications (cf.\ \cite{Mukund12a}), we believe that our procedure can be of interest in other contexts, too, and be used to simplify or even generalize existing results.

\smallskip

To summarize, the motivation of this work comes from distributed algorithms and formal methods.
On the one hand, we tackle an important problem from distributed computing.
On the other hand, our results shed some light on the expressive power of message-passing systems.
In fact, previous logical studies of \CFMs with unbounded FIFO channels
in terms of existential MSO logic (without happens-before relation and, respectively, restricted to two first-order variables) and propositional dynamic logic \cite{BolligJournal,BKM-lmcs10,BFG-stacs18} do not allow us to solve the gossip problem or to show that CFMs capture abovementioned linear-time temporal logics.

\subparagraph{Outline.}

The paper is structured as follows:
In Section~\ref{sec:prel}, we define communicating finite-state machines (CFMs), a fundamental model of message-passing systems. The gossip problem is introduced in Section~\ref{sec:gossip}. Our (non-deterministic) solution to the gossip problem is distributed over two parts, Sections~\ref{sec:paths} and \ref{sec:preorder-cfm}. In fact, it is obtained as an instance of a more general approach, in which we are able to compare the latest information transmitted along paths described by path expressions. This general solution finally allows us to translate formulas from linear-time temporal logic into CFMs (Section~\ref{sec:logic}). We conclude in Section~\ref{sec:conclusion}.

% ------------------------------------------------------------
% ------------------------------------------------------------
% ------------------------------------------------------------

\section{Preliminaries}\label{sec:prel}

\subparagraph{Communicating Finite-State Machines.}

We consider a distributed system with a fixed finite set of processes $\Procs$.
Processes are connected in a communication network that contains a FIFO channel
from every process $p$ to any other process $q$ such that $p \neq q$.
We also assume a finite set $\Sigma$ of \emph{labels}, which provide information
about events in a system execution such as ``enter critical region'' or ``output some value''.

In a communicating finite-state machine, each process $p \in \Procs$ can
perform local actions, or send/receive messages from a finite set of messages $\Msg$.
Process $p$ is represented as a finite transition system
$\A_p = (S_p,\init_p,\Delta_p)$ where $S_p$ is the finite set of (local) states,
$\init_p \in S_p$ is the initial state, and $\Delta_p$ is the transition relation.

A transition in $\Delta_p$ is of the form $t=(s,\act,s')$ where $s,s' \in S_p$ are the source
state and the target state, referred to as $\source(t)$ and $\target(t)$,
respectively.  Moreover, $\act$ determines the effect of $t$.  First, $\act$ may
be of the form $\loc{a}$ with $a \in \Sigma$.
In that case, $t$ performs a local computation that does not involve any communication primitive. We let $\tlabel(t) = a$.
Second, $\act$ may be of the form $\send{a}{\msg}{q}$. Then, in addition to performing $a \in \Sigma$, process $p$ sends message $\msg \in \Msg$ to process $q \in \Procs \setminus \{p\}$.
More precisely, $\msg$ is placed in the FIFO channel from $p$ to $q$.
We let $\receiver(t) = q$, $\tmsg(t) = \msg$, and $\tlabel(t) = a$.
Finally, if $\act = \rec{a}{\msg}{q}$, then $p$ receives message $\msg$ from $q$,
and we let $\sender(t) = q$, $\tmsg(t) = \msg$, and $\tlabel(t) = a$.

In addition, our system is equipped with an acceptance condition.  In order for
an execution to be accepting, all channels have to be empty and the collection
of local states in which processes terminate must belong to a set $\Acc
\subseteq \prod_{p \in \Procs} S_p$.  We call the tuple $\C = ((\A_p)_{p \in
\Procs},\Msg,\Acc)$ a \emph{communicating finite-state machine (CFM)} over $\Procs$ and $\Sigma$.

\begin{example}\label{ex:cfm}
Consider the simple CFM depicted in Figure~\ref{fig:cfm}.
The set of processes is $\Procs = \{p,q,r\}$.
Moreover, we have $\Sigma = \{\bbullet,\abullet,\nbullet\}$ and $\Msg = \{\bbullet,\abullet\}$.
Process $p$ sends messages to $q$ and $r$. Each message can be either
$\bbullet$ or $\abullet$, and the message sent is made ``visible'' in terms of $\Sigma$.
Process $r$ simply forwards every message it receives to $q$.
In any case, the action is $\nbullet$, which means that we do not want to reason about the forwarding itself.
Finally, $q$ receives and ``outputs'' messages from $p$ and $r$ in any order. Note that, in this example, there are no local transitions,
i.e., every transition is either sending or receiving.
\end{example}

\usetikzlibrary{positioning,automata}

\begin{figure}[h]
\centering
\begin{tikzpicture}[shorten >=1pt,node distance=2cm,on grid, semithick,>=stealth]
\tikzstyle{every state}=[draw=black!50,very thick,fill=white, circle, minimum size=1pt]
  \node[state,initial,accepting, initial text=$p$, initial above]   (q_0) at (0,0)               {$s_0^p$};
  \path[->] (q_0) edge [loop right] node [right]
    {$\!\!\!\!\begin{array}{l}\send{\bbullet}{\bbullet}{q}\\%
                              \send{\bbullet}{\bbullet}{r}\\%
                              \send{\abullet}{\abullet}{q}\\%
                              \send{\abullet}{\abullet}{r}\\%
              \end{array}$} (q_0);

  \node[state,initial,accepting, initial text=$r$, initial above]   (q_0) at (5.8,0)               {$s_0^r$};
  \node[state]           (q_1) [left=of q_0] {$s_1^r$};
  \node[state]           (q_2) [right=of q_0] {$s_2^r$};
  \path[->] (q_0) edge [bend right=30] node [above] {$\rec{\nbullet}{\bbullet}{p}$} (q_1)
            (q_1) edge [bend right=30] node [below] {$\send{\nbullet}{\bbullet}{q}$} (q_0)
            (q_0) edge [bend left=30] node [above] {$\rec{\nbullet}{\abullet}{p}$} (q_2)
            (q_2) edge [bend left=30] node [below] {$\send{\nbullet}{\abullet}{q}$} (q_0);
              
  \node[state,initial,accepting, initial text=$q$, initial above]   (q_0) at (10,0)               {$s_0^q$};
  \path[->] (q_0) edge [loop right] node [right]
    {$\!\!\!\!\begin{array}{l}\rec{\bbullet}{\bbullet}{p}\\%
                              \rec{\bbullet}{\bbullet}{r}\\%
                              \rec{\abullet}{\abullet}{p}\\%
                              \rec{\abullet}{\abullet}{r}\\%
              \end{array}$} (q_0);
\end{tikzpicture}
\caption{A communicating finite-state machine\label{fig:cfm}}
\end{figure}
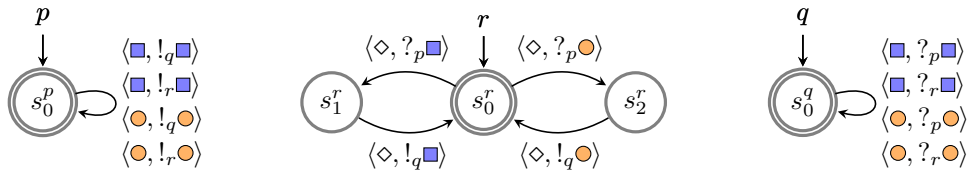

\subparagraph{Message Sequence Charts.}

An execution of $\C$ can be described by a diagram as depicted in Figure~\ref{fig:msc}.
Process $p$ performs eight transitions, alternately sending a message to $q$ and $r$.
Note that the execution does not keep track of states and messages (unless made ``visible'' by means of $\Sigma$).

Let us describe a structure like in Figure~\ref{fig:msc} formally.
We have a nonempty finite set $E$ of \emph{events} (in the example, $E=\{e_0,\ldots,e_7,g_0,\ldots,g_7,f_0,\ldots,f_7\}$).
With each event, we associate its process and an action from $\Sigma$, i.e., we have mappings $\ploc: E \to \Procs$ and $\lab: E \to \Sigma$.
We let $E_p \df \{e \in E \mid \ploc(e) = p\}$ be the set of events executed by process $p$.
A binary relation ${\prel} \subseteq E \times E$ connects consecutive events of a process:
For all $(e,f) \in {\prel}$, there is $p \in \Procs$ such that both $e$ and $f$ are in $E_p$.
Moreover, for all $p \in \Procs$, ${\prel} \cap (E_p \times E_p)$ is the direct successor relation of some total order on $E_p$.
Finally, the message relation ${\mrel} \subseteq E \times E$ connects a pair of events that represent a message exchange.
We require that
\begin{itemize}
\item every event belongs to at most one pair from $\mrel$, and

\item for all $(e,f),(e',f') \in {\mrel}$ such that $e,e' \in E_p$ and $f,f' \in E_q$, we have
both $p \neq q$ and (FIFO) $e \prel^\ast e'$ iff $f \prel^\ast f'$.
\end{itemize}
Finally, ${\le} \df ({\prel} \cup {\mrel})^\ast$ must be a partial order.
Its strict part is denoted ${<} = ({\prel} \cup {\mrel})^+$.

We call $M=(E,\prel,\mrel,\ploc,\lambda)$ a \emph{message sequence chart (MSC)} over $\Procs$ and $\Sigma$. The set of message sequence charts is denoted by $\MSCs{\Procs}{\Sigma}$.

\begin{example}\label{ex:msc}
Let us come back to the MSC from Figure~\ref{fig:msc}.
We have $\ploc(e_2) = p$, $\lambda(e_2) = \abullet$, $\lambda(f_2) = \bbullet$,
 and $\lambda(g_i) = \nbullet$ for all $i \in \{0,\ldots,7\}$.
The process relation restricted to $p$ is $e_0 \prel e_1 \prel \ldots \prel e_7$.
We also have $g_0 \prel g_1 \prel \ldots$ and $f_0 \prel f_1 \prel \ldots$
Concerning the message relation, $e_4 \mrel f_5$ and $e_7 \mrel g_6$, among others.
\end{example}

\tikzstyle{acirc} = [draw, fill, circle, inner sep=0, minimum size=0.25cm, orange!60, draw=black]
\tikzstyle{bcirc} = [draw, fill, rectangle, inner sep=0, minimum size=0.2cm, blue!50, draw=black]
\tikzstyle{ncirc} = [draw, fill=white, diamond, inner sep=0, minimum size=0.3cm]

\begin{figure}[h]
\centering
    \begin{tikzpicture}[semithick,>=stealth]
      \draw[->] (-0.25,0) -- (12.5,0);
      \draw[->] (-0.25,1) -- (12.5,1);
      \draw[->] (-0.25,2) -- (12.5,2);

     \node[bcirc,label=above:$e_0$] (e0) at (0.25,2) {};
     \node[bcirc,label=below:$f_0$] (f0) at (0.25,0) {};
     \draw[->] (e0) -- (f0);

     \node[bcirc,label=above:$e_1$] (e1) at (1.125,2) {};
     \node[ncirc,label=below:$g_0$] (g0) at (1.125,1) {};
     \draw[->] (e1) -- (g0);

     \node[acirc,label=above:$e_2$] (e2) at (2,2) {};
     \node[acirc,label=below:$f_1$] (f1) at (2,0) {};
     \draw[->] (e2) -- (f1);

     \node[ncirc,label=above:$g_1$] (g1) at (3,1) {};
     \node[bcirc,label=below:$f_2$] (f2) at (3.75,0) {};
     \draw[->] (g1) -- (f2);

     \node[bcirc,label=above:$e_3$] (e3) at (4,2) {};
     \node[ncirc,label=below:$g_2$] (g2) at (4,1) {};
     \draw[->] (e3) -- (g2);

     \node[ncirc,label=above:$g_3$] (g3) at (5,1) {};
     \node[bcirc,label=below:$f_3$] (f3) at (5.5,0) {};
     \draw[->] (g3) -- (f3);
     
     \node[bcirc,label=above:$e_4$] (e4) at (5,2) {};
     \node[bcirc,label=below:$f_5$] (f5) at (9,0) {};
     \draw[->] (e4) -- (f5);

     \node[acirc,label=above:$e_5$] (e5) at (6,2) {};
     \node[ncirc,label=below:$g_4$] (g4) at (6,1) {};
     \draw[->] (e5) -- (g4);

     \node[ncirc,label=above:$g_5$] (g5) at (8,1) {};
     \node[acirc,label=below:$f_4$] (f4) at (8,0) {};
     \draw[->] (g5) -- (f4);

     \node[bcirc,label=above:$e_6$] (e6) at (9,2) {};
     \node[bcirc,label=below:$f_6$] (f6) at (10.5,0) {};
     \draw[->] (e6) -- (f6);

     \node[acirc,label=above:$e_7$] (e7) at (11,2) {};
     \node[ncirc,label=below:$g_6$] (g6) at (11,1) {};
     \draw[->] (e7) -- (g6);

     \node[ncirc,label=above:$g_7$] (g7) at (12,1) {};
     \node[acirc,label=below:$f_7$] (f7) at (12,0) {};
     \draw[->] (g7) -- (f7);

      \node at (-0.5,0) {$q$};
      \node at (-0.5,1) {$r$};
      \node at (-0.5,2) {$p$};            
    \end{tikzpicture}
    \caption{A message sequence chart\label{fig:msc}}
  \end{figure}
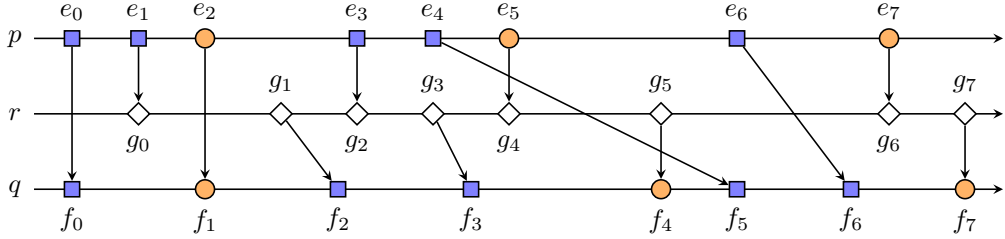

\subparagraph{Runs and the Language of a CFM.}

Let $\C = ((\A_p)_{p \in \Procs},\Msg,\Acc)$ be a CFM and $M=(E,\prel,\mrel,\ploc,\lambda)$ be an MSC over $\Procs$ and $\Sigma$. A \emph{run} of $\C$ on $M$ associates with
every event $e \in E_p$ ($p\in\Procs$) the transition $\rho(e) \in\Delta_p$ that
is executed at $e$.  We require that
\begin{enumerate}
\item for all events $e \in E$, we have $\tlabel(\rho(e)) = \lambda(e)$,
\item for all processes $p \in \Procs$ such that $E_p \neq \emptyset$, we have $\source(\rho(e)) = \init_p$ where $e$ is the first event of $p$ (i.e., $e$ does not have a $\prel$-predecessor),
\item for all process edges $(e,f) \in {\prel}$, we have $\target(\rho(e)) = \source(\rho(f))$,
\item for all local events $e \in E$ ($e$ is neither a send nor a receive), $\rho(e)$ is a local transition, and

\item for all message edges $(e,f) \in {\mrel}$, say, with $e \in E_p$ and $f \in
E_q$, $\rho(e)\in\Delta_p$ is a send transition and $\rho(f)\in\Delta_q$ is a
receive transition such that $\tmsg(\rho(e)) = \tmsg(\rho(f))$,
$\receiver(\rho(e)) = q$, and $\sender(\rho(f)) = p$.
\end{enumerate}

To determine whether $\rho$ is accepting, we collect the last state of every
process $p$.  If $E_p \neq \emptyset$, then let $s_p$ be $\target(\rho(e))$
where $e$ is the last event of $E_p$.  Otherwise, let $s_p = \init_p$.  Now,
$\rho$ is said to be \emph{accepting} if $(s_p)_{p \in \Procs} \in \Acc$.

Finally, the \emph{language} of $\A$ is $L(\A) \df \{M \in \MSCs{\Procs}{\Sigma} \mid$ there is an accepting run of $\A$ on $M\}$.
For example, the MSC from Figure~\ref{fig:msc} is in the language of the \CFM from Figure~\ref{fig:cfm}.

% ------------------------------------------------------------
% ------------------------------------------------------------
% ------------------------------------------------------------

\section{The Gossip Problem}\label{sec:gossip}

We are looking for a protocol (a CFM) that solves the gossip problem: When a
process $q$ receives a message at some event $f \in E_q$, it should be able to tell
what the most recent information is that it has on another process, say $p$.  More
precisely, it should determine the label $\lambda(e)$ of the last (i.e., most
recent) event $e$ of $E_p$ that is in the (strict) past of $f$.  For example,
consider the MSC in Figure~\ref{fig:gossip-msc} (for the moment, we ignore the bottom
part of the figure).  At the time of executing event
$f_5$, process $q$ is supposed to ``output'' $\abullet$, since the most recent
event on $p$ is $e_5$.

Let us formally define what it means to be the \emph{most recent} event.
For all $f \in E$ and $p \in P$, we define
$\rPastp p f  = \{e \in E_p \mid e < f\}$
to be the set of events on process $p$ that are in the \emph{past} of $f$.
We let
\[
  \lastp p f = \begin{cases}
    \max (\rPastp p f)& \text{if } \rPastp p f \neq \emptyset \\
    \botp p & \text{otherwise} \, .
  \end{cases}
\]
Thus, $\lastp p f$ is the most recent event of $p$ in the past of $f$.

\begin{example}
Consider the MSC from Figure~\ref{fig:gossip-msc}. We have
$\rPastp p {f_5} = \{e_0,\ldots,e_5\}$ and, therefore, $\lastp p {f_5} = e_5$.
Moreover, $\lastp p {f_2} = e_2$.
\end{example}

The CFM from Figure~\ref{fig:cfm} (cf.\ Example~\ref{ex:cfm}) can be seen as a first (na{\"i}ve) attempt to solve the gossip problem. When $q$ receives a message from $p$, it ``outputs'' the color of the sending event, and when $q$ receives a message from $r$, it outputs the color transmitted by $r$.
However, both rules are erroneous: Consider the MSC in Figure~\ref{fig:msc}.
At $f_2$ and $f_5$, process $q$ should have announced $\abullet$, but it outputs $\bbullet$. Actually, what we would like to have is the behavior depicted in Figure~\ref{fig:gossip-msc} where, for all $i \in \{0,\ldots,7\}$, we get $\lambda(f_i) = \lambda(\lastp p {f_i})$.

\begin{figure}[h]
\centering
    \begin{tikzpicture}[semithick,>=stealth]
      \draw[->] (-0.25,0) -- (12.5,0);
      \draw[->] (-0.25,1) -- (12.5,1);
      \draw[->] (-0.25,2) -- (12.5,2);

     \node[bcirc,label=above:$e_0$] (e0) at (0.25,2) {};
     \node[bcirc,label=below:$f_0$] (f0) at (0.25,0) {};
     \draw[->] (e0) -- (f0);

     \node[bcirc,label=above:$e_1$] (e1) at (1.125,2) {};
     \node[ncirc,label=below:$g_0$] (g0) at (1.125,1) {};
     \draw[->] (e1) -- (g0);

     \node[acirc,label=above:$e_2$] (e2) at (2,2) {};
     \node[acirc,label=below:$f_1$] (f1) at (2,0) {};
     \draw[->] (e2) -- (f1);

     \node[ncirc,label=above:$g_1$] (g1) at (3,1) {};
     \node[acirc,label=below:$f_2$] (f2) at (3.75,0) {};
     \draw[->] (g1) -- (f2);

     \node[bcirc,label=above:$e_3$] (e3) at (4,2) {};
     \node[ncirc,label=below:$g_2$] (g2) at (4,1) {};
     \draw[->] (e3) -- (g2);

     \node[ncirc,label=above:$g_3$] (g3) at (5,1) {};
     \node[bcirc,label=below:$f_3$] (f3) at (5.5,0) {};
     \draw[->] (g3) -- (f3);
     
     \node[bcirc,label=above:$e_4$] (e4) at (5,2) {};
     \node[acirc,label=below:$f_5$] (f5) at (8.75,0) {};
     \draw[->] (e4) -- (f5);

     \node[acirc,label=above:$e_5$] (e5) at (6,2) {};
     \node[ncirc,label=below:$g_4$] (g4) at (6,1) {};
     \draw[->] (e5) -- (g4);

     \node[ncirc,label=above:$g_5$] (g5) at (7.5,1) {};
     \node[acirc,label=below:$f_4$] (f4) at (7.5,0) {};
     \draw[->] (g5) -- (f4);

     \node[bcirc,label=above:$e_6$] (e6) at (9,2) {};
     \node[bcirc,label=below:$f_6$] (f6) at (10.5,0) {};
     \draw[->] (e6) -- (f6);

     \node[acirc,label=above:$e_7$] (e7) at (11,2) {};
     \node[ncirc,label=below:$g_6$] (g6) at (11,1) {};
     \draw[->] (e7) -- (g6);

     \node[ncirc,label=above:$g_7$] (g7) at (12,1) {};
     \node[acirc,label=below:$f_7$] (f7) at (12,0) {};
     \draw[->] (g7) -- (f7);

      \node at (-0.5,0) {$q$};
      \node at (-0.5,1) {$r$};
      \node at (-0.5,2) {$p$};
      
      \node at (2,-1.2) {\scalebox{0.75}{$\begin{array}{c}\fa{\p}{{\pamove}\p'}(f_1)\\=f_3\end{array}$}};
       \node at (3.75,-1.2) {\scalebox{0.75}{$\begin{array}{c}\fa{\p}{{\pamove}\p'}(f_2)\\=f_3\end{array}$}};
      \node at (5.5,-1.2) {\scalebox{0.75}{$\begin{array}{c}\fa{\p}{{\pamove}\p'}(f_3)\\=f_3\end{array}$}};
      \node at (7.2,-1.2) {\scalebox{0.75}{$\begin{array}{c}\fa{\p'}{{\ppmove}\p}(f_4)\\=f_6\end{array}$}};
      \node at (8.8,-1.2) {\scalebox{0.75}{$\begin{array}{c}\fa{\p'}{{\ppmove}\p}(f_5)\\=f_6\end{array}$}};
      \node at (10.4,-1.2) {\scalebox{0.75}{$\begin{array}{c}\fa{\p'}{{\ppmove}\p}(f_6)\\=f_6\end{array}$}};
      \node at (12,-1.2) {\scalebox{0.75}{$\begin{array}{c}\fa{\p}{{\pamove}\p'}(f_7)\\=f_7\end{array}$}};

      \node at (0.25,-2) {\scalebox{0.9}{$\p' \pless_{f_0} \p$}};
      \node at (2,-2) {\scalebox{0.9}{$\p' \pless_{f_1} \p$}};
      \node at (3.75,-2) {\scalebox{0.9}{$\p' \pless_{f_2} \p$}};
      \node at (5.5,-2) {\scalebox{0.9}{$\p \ple_{f_3} \p'$}};
      \node at (7.5,-2) {\scalebox{0.9}{$\p \ple_{f_4} \p'$}};
      \node at (9,-2) {\scalebox{0.9}{$\p \ple_{f_5} \p'$}};
      \node at (10.5,-2) {\scalebox{0.9}{$\p' \pless_{f_6} \p$}};
      \node at (12,-2) {\scalebox{0.9}{$\p \ple_{f_7} \p'$}};

      \node at (0.25,-2.5) {$\underbrace{\hspace{3.7em}}_{\textup{Lemma}~\ref{ple-char}.\ref{lem:initle}}$};
      \node at (3.75,-2.5) {$\underbrace{\hspace{13.7em}}_{\textup{Lemma}~\ref{ple-char}.\ref{lem:gtle}}$};
      \node at (9,-2.5) {$\underbrace{\hspace{12.2em}}_{\textup{Lemma}~\ref{ple-char}.\ref{lem:legt}}$};
      \node at (12.,-2.5) {$\underbrace{\hspace{3.7em}}_{\textup{Lemma}~\ref{ple-char}.\ref{lem:gtle}}$};
                  
    \end{tikzpicture}
    \caption{Comparison of $\p = {\mmove{p}{q}}{\pamove}$ and $\p' = {\mmove{p}{r}}{\pamove}{\mmove{r}{q}}{\pamove}$\label{fig:gossip-msc}}
  \end{figure}

Formally, we will treat ``outputs'' in terms of additional labels from another finite alphabet $\Xi$.
To do so, we consider \CFMs and MSCs over $\Procs$ and $\Sigma \times \Xi$. An MSC over $\Procs$ and $\Sigma \times \Xi$ is called an \emph{extended MSC}. It can be interpreted, in the expected way, as a pair $(M,\xi)$ where $M=(E,\prel,\mrel,\ploc,\lambda)$ is an MSC over $\Procs$ and $\Sigma$, and $\xi: E \to \Xi$. If $(M,\xi)$ is accepted by the gossip \CFM, $\xi(e)$ shall provide the latest information that $e$ has about any other process. That is, $\Xi$ is the finite set of functions from $\Procs$ to $\Sigma \cup \{\botp p\}$. We assume $\bot \not\in \Sigma$ and $\lambda(\bot) = \bot$.

We are now looking for a \CFM $\Agossip$ over $\Procs$ and $\Sigma \times \Xi$ that has the following property:
\begin{center}
\fbox{\parbox{0.85\textwidth}{\it
The language $L(\Agossip)$ is the set of extended MSCs
$((E,\prel,\mrel,\ploc,\lambda),\xi)$ such that, for all events $e \in E$,
$\xi(e)$ is the function from $\Procs$ to $\Sigma \uplus \{\bot\}$ defined
by $\xi(e)(p) = \lambda(\lastp p e)$.
}\;}
\end{center}

Thus, the gossip CFM $\Agossip$ allows a process to infer, at any time,
the most recent information that it has about all other processes wrt.\ the \emph{causal past}.
In fact, we will pursue a more general approach based on \emph{path expressions}.
A path expression allows us to define what we actually mean by ``causal past''.
More precisely, it acts as a filter that considers only events in the past that are 
(co-)reachable via certain paths (e.g., visiting only certain processes or at least one event with a given label).
Path expressions and their properties are studied in Section~\ref{sec:paths}.
In Section~\ref{sec:preorder-cfm}, we construct a CFM
that, at any event, is able to tell which of two path expressions provides more recent information.
We then obtain $\Agossip$ as a corollary.

\section{Comparing Path Expressions}\label{sec:comp}\label{sec:paths}

In this section, we introduce path expressions and establish some of their properties.

\subsection{Path Expressions}

Let us again look at our running example (cf.\ Figure~\ref{fig:gossip-msc}).
In the gossip problem, we need to know whether the most
recent information has been provided along a message from $p$ to $q$,
which will be represented by the path expression $\p = {\mmove{p}{q}}{\pamove}$, or via the intermediate process $r$,
represented by the path expression $\p' = {\mmove{p}{r}}{\pamove}{\mmove{r}{q}}{\pamove}$.
We will write $\p \ple_{f_5} \p'$
to describe the fact that $\last \p {f_5} \le \last {\p'} {f_5}$,
where $\last \p {f_5}=e_4$ and $\last {\p'} {f_5}=e_5$
denote the most recent events from which a $\p$-path and, respectively,
$\pi'$-path to $f_5$ exist.

Let us be more formal.  A path expression is simply a finite word over the
alphabet
$
  \moves = \{ \pmove, \pamove \} \cup \{ \mmove p q \mid p,q \in P$, $p \neq q \}
  \cup \{\amove a \mid a \in \Sigma\}
$.
We let $\varepsilon$ be the empty word and introduce $\ppmove$ as a macro for the word ${\pmove}{\pamove}$.
Let $M = (E,\prel,\mrel,\ploc,\lambda)$ be an MSC. For all path expressions $\p \in \moves^\ast$, we define a relation $\sem M \p \subseteq E \times E$ as follows:\\
{
\begin{minipage}[b]{0.3\textwidth}
\begin{align*}
  \sem M {\varepsilon} & = \{ (e,e) \mid e \in E \} \\
  \sem M {\amove a} & = \{ (e,e) \in E \times E \mid \lambda(e) = a \} \\
  \sem M {\mmove p q} & = \{ (e,f) \in E_p \times E_q \mid e \mrel f \}
\end{align*}
\end{minipage}
\begin{minipage}[b]{0.3\textwidth}
\begin{align*}
  \sem M {\pmove} & = \{ (e,f) \in E \times E \mid e \prel f \} \\
  \sem M {\pamove} & = \{ (e,f) \in E \times E \mid e \pamove f \}
  \end{align*}
\end{minipage}\vspace{-0.6ex}
\begin{align*}
  \sem M {\p\p'} & = \sem{M}{\p}\circ\sem{M}{\p'}
                   = \{ (e,g) \in E \times E \mid \exists f \in E:
                        (e,f)\in\sem{M}{\p} \land (f,g)\in\sem{M}{\p'} \} \, .
\end{align*}
}
\begin{example}
  Consider the MSC $M$ from Figure~\ref{fig:gossip-msc}.
  For $\p = {\mmove{p}{q}}{\pamove}$ and $\p' = {\mmove{p}{r}}{\pamove}{\mmove{r}{q}}{\pamove}$,
  we have  $(e_4,f_5) \in \sem M {\p}$ and $(e_5,f_5) \in  \sem M {\p'}$.
  Moreover, $\sem M {\bbullet\pmove\bbullet{\mmove{p}{q}}} = \{(e_3,f_5)\}$.
\end{example}

We say that a pair of processes $(p,q)$ is \emph{compatible} with $\p \in
\moves^\ast$ if $\p$ may describe a path from $p$ to $q$.  Formally, we define
$\Comp{\p} \subseteq \Procs \times \Procs$ inductively as follows:
$\Comp{\varepsilon} = \Comp{\amove a} = \Comp{\pmove} = \Comp{\pamove} = \{(p,p) \mid p \in \Procs\}$, $\Comp{\mmove p q} = \{(p,q)\}$,
and $\Comp{\p \p'} = \Comp{\p} \circ \Comp{\p'}$, where $\circ$ denotes the usual product of binary relations.
Note that, for each $p$, there is at most one $q$ such that $(p,q) \in
\Comp{\p}$.  Conversely, for each $q$, there is at most one $p$ such that $(p,q)
\in \Comp{\p}$.  We denote by $\Paths_{p,q}$ the set of path expressions 
$\p\in\moves^*$ such that $(p,q)\in\Comp{\p}$.

\begin{example}
  We have $\Comp{{\mmove{p}{r}}{\pamove}{\mmove{r}{q}}{\pamove}} = \{(p,q)\}$,
  $\Comp{{\mmove{p}{q}}{\pamove}{\mmove{q}{p}}} = \{(p,p)\}$,
  $\Comp{\bbullet\pmove\bbullet{\mmove{p}{q}}} = \{(p,q)\}$, and
   $\Comp{{\mmove{p}{q}}{\pamove}{\mmove{r}{p}}} = \emptyset$.
\end{example}

Next, given $\p \in \moves^\ast$ and $e \in E$, we define $\last \p e$ and $\first{\p}{e}$, which denote the most recent (resp.\ very next) event from which there is a $\p$-path to $e$
(resp.\ to which there is a $\p$-path from $e$). We extend $\le$ with the new elements $\bot$ and $\top$ by setting $\bot < e < \top$ for all $e \in E$. As before, we will assume $\lambda(\bot) = \bot$. Moreover, $\lambda(\top) = \top$.

All events $f$ such that $\Pto M f \p
e$ (resp.\ $\Pto M e \p f$) are located on the same process. Hence, we can 
define, with $\max\emptyset=\bot$ and $\min\emptyset = \top$:
\begin{align*}
  \last{\p}{e} & =\max\,\sem{M}{\p}^{-1}(e)
  = \max \{ f \in E \mid \Pto M f \p e \} \\
  \first{\p}{e} & =\min\,\sem{M}{\p}(e)
  = \min \{ f \in E \mid \Pto M e \p f \} \,.
\end{align*}

The next lemma states that $\lastf \p$ and $\firstf \p$ are monotone.

\begin{lemma}\label{lem:monotone}
  Let $\p \in \moves^\ast$ and $e,f \in E$. The following hold:
  \begin{enumerate}
    \item If $\last \p e \neq \botp p$, $\last{\p}{f} \neq \botp p$, and $e\pamove f$, then $\last \p e \le \last \p f$.
    
    \item If $\first \p e \neq \topp q$, $\first{\p}{f} \neq \topp q$, and $e\pamove f$, then $\first \p e \le \first \p f$.

    \item If $\last{\p}{e}\neq\botp{p}$, then $\last{\p{\pamove}}{e}=\last{\p}{e}$.

    \item If $\first{\p}{e}\neq\topp{q}$, then $\first{{\pamove}\p}{e}=\first{\p}{e}$.
  \end{enumerate}
\end{lemma}

\begin{proof}
  We show 1.\ and 3. The other two cases are analogous.
  For 1., the proof is by induction on $\p$.  We assume $\last \p e \neq \bot$
  and $\last{\p}{f} \neq \bot$.  The case $\p = \varepsilon$ is immediate.

  Suppose $\p = \p' \mmove {r} q$. There exists
  some $e' \in E_{r}$ such that $e' \mrel e$ and
  $\last \p e = \last {\p'} {e'}$.
  Similarly, there exists $f' \in E_{r}$ such that $f' \mrel f$ and
  $\last \p f = \last {\p'} {f'}$.
  Because of the FIFO ordering, we have $e' \pamove f'$, and by induction
  hypothesis, we get $\last \p e \le \last \p f$.
  
  The cases $\p = \p' {\pmove}$ and $\p = \p' \amove a$ are similar.

  Suppose $\p=\p'{\pamove}$. Due to $(\last{\p}{e},e)\in\sem{M}{\p}$ and $e\pamove f$,
  we have $(\last{\p}{e},f)\in\sem{M}{\p}$.
  By definition of $\last \p f$, we then get $\last \p e \le \last \p f$.
  
  For 3., we assume that $\last{\p}{e}\neq\botp{p}$. We have 
  $\sem{M}{\p}\subseteq\sem{M}{\p{\pamove}}$ hence we get 
  $g=\last{\p}{e}\leq\last{\p{\pamove}}{e}=g'$. Now, there is $e'$ such that 
  $g'=\last{\p}{e'}$ and $e'\pamove e$. From 1., we deduce that $g'\leq g$.
\end{proof}

Now, let us define formally when a path $\p'$ provides (strictly) more recent
information than a path $\p$.  Fix $p,q \in P$.
For all $e \in
E_q$ and $\p, \p' \in \Paths_{p,q}$, we let
\begin{align*}
  \p \ple_e \p'
  &\qquad\text{if}\qquad
    \last \p e \le \last {\p'} e
  \\
  \p \pless_e \p'
  &\qquad\text{if}\qquad  \last \p e < \last {\p'} e,\ \text{i.e., }
    \p' \not \ple_e \p
      \, .
\end{align*}
The definition is illustrated in Figure~\ref{fig:gossip-msc}.

Recall that our goal is to construct a \CFM computing the label of $\lastp p e$
for all events $e \in E_q$.  Later (in Section~\ref{sec:cfm-labels}), we show that,
for all $\p$, there exists a \CFM associating with each event $e$ the label
of $\last \p e$.
Thus, it will be enough to construct a \CFM that identifies, for each event
$e$, some $\p \in \moves^\ast$ such that $\last \p e = \lastp p e$.
Moreover, path expressions of bounded length will suffice: If $f < e$, then
there is a path from $f$ to $e$ that enters and leaves each process at most once. 

To achieve our goal, we will build a \CFM $\Ale$ computing the total preorders
$\ple_e$ (restricted to path expressions of bounded size) for all events $e$ on a given process $q$.  In particular, $\Ale$ is
sufficient to determine, for all $e\in E_q$ and $p \in \Procs$, some 
$\p\in\Paths_{p,q}$ such that $\lastp p e = \last \p e$.  The idea is that
$\Ale$ first determines $\ple_e$ for the minimal event $e$ in $E_q$. Then, for
all $\p,\p' \in \Paths_{p,q}$, it computes the set of events where the order
between $\p$ and $\p'$ is switched. In Figure~\ref{fig:gossip-msc}, these
switching events are $f_3$, $f_6$, and $f_7$.  The next subsection provides a
characterization of the preorder that can then (in
Section~\ref{sec:gossip-cfm}) be implemented as a CFM.

\subsection{A Characterization of $\ple_e$}

Given $p,q\in\Procs$ and $\p,\p' \in \Paths_{p,q}$, we define the function
$\fa \p {\p'}\colon E_q \to E_q \cup \{\bot,\top\}$ (omitting index $(p,q)$)
as follows: $\fa \p {\p'} (e) = \first {\p'} {\last \p e}$, with
$\first {\p'} {\bot} = \bot$.
So we have $\fa \p {\p'} (e) = f \in E_q$ if there is $g\in E_p$ such that
$\last \p e = g$ and $\first {\p'} g = f$,
$\fa \p {\p'} (e) = \bot$ if $\last \p e = \bot$,
and $\fa \p {\p'} (e) = \top$ if $\last \p e = g \in E_p$ but
$\first {\p'} g = \top$.

From Lemma~\ref{lem:monotone}, we can deduce monotonicity of $\fa \p {\p'}$:

\begin{lemma}\label{lem:monotone2}
  Suppose $e \pamove f$ and $\fa \p {\p'} (e),\fa \p {\p'} (f) \in E_q$.
  Then, $\fa \p {\p'} (e) \le \fa \p {\p'} (f)$.
\end{lemma}

\begin{example}
  Consider, again, Figure~\ref{fig:gossip-msc} with $\p =
  {\mmove{p}{q}}{\pamove}$ and $\p' =
  {\mmove{p}{r}}{\pamove}{\mmove{r}{q}}{\pamove}$.  We get $\fa \p
  {{\pamove}\p'} (f_3) = f_3$ and $\fa \p {{\pamove}\p'} (f_5) = f_4$. 
  Since $\last {\p'} {f_0} = \bot$ and $\last {\p} {f_0}=e_0 \neq \bot$,
  we have $\p' \pless_{f_0} \p$.
\end{example}

\begin{figure}
\centering
  \begin{minipage}[b]{0.18\textwidth}
    \scalebox{0.85}{
    \begin{tikzpicture}[scale=0.8,font=\footnotesize,inner sep=2pt,auto,>=stealth]
      \draw (0,0) -- (2,0);
      \draw (0,2) -- (2,2);

      \node (q) at (-0.3,0) {$q$};
      \node (p) at (-0.3,2) {$p$};
    
      \node[dot,label=below:{$f$}] (e) at (0.5,0) {};
      \node[dot] (g) at (0.5,2) {};
      \node[dot] (h) at (1.7,2) {};
      \node[anchor=west] at ($(g) + (-0.2,0.3)$) {$g ~~\le~~ g'$};

      \path
      (g) edge[decorate,decoration={snake,post length=1mm,amplitude=0.6mm},->] node[left] {$\p$} (e)
      (h) edge[decorate,decoration={snake,post length=1mm,amplitude=0.6mm},->] node[right] {$\p'$} (e);

      \path[purple,very thick,shorten > =5pt, shorten < = 5pt,->]
      (e) edge[bend left=40] node[left] {$\mathsf{pred}_\p$} (g)
      (g) edge[bend left=80,looseness=2] node[right] {$\mathsf{succ}_{{\pamove}\p'}$} (e);
    \end{tikzpicture}}
  \end{minipage}
  \begin{minipage}[b]{0.38\textwidth}
        \scalebox{0.85}{
    \begin{tikzpicture}[scale=0.8,font=\footnotesize,inner sep=2pt,auto,>=stealth]
      \draw (0,0) -- (6,0);
      \draw (0,2) -- (6,2);

      \node (q) at (-0.3,0) {$q$};
      \node (p) at (-0.3,2) {$p$};

      \node[anchor=west] at (-0.3,2.4)
      {$\last {\p'{\pamove}} e ~<~ \last {\p{\pamove}} e$};
      \node[dot] (p) at (0.5,2) {};
      \node[dot] (p') at ($(p)+(2.5,0)$) {};
      \node[dot] (g') at (4.7,2) {};
      \node[dot,right=0.8cm of g'] (g) {};
      \node[above=0.1cm of g'] {$g\vphantom{g'}$};
      \node[above right= 0.1cm and 0.2cm of g'] {$\vphantom{g'}{\le}$};
      \node[above=0.1cm of g] {$g'$};
      \node[dot,label={[name=el]below:{$e$\vphantom{$f$}}}] (e) at (4.3,0) {};
      \node[dot,label={[name=fl]below:{$f$}}] (f) at (5.8,0) {};

      \path
      (p) edge[decorate,decoration={snake,post length=1mm,amplitude=0.6mm},->]
      node[left,xshift=-5pt] {$\p'{\pamove}$} (e)
      (p') edge[decorate,decoration={snake,post length=1mm,amplitude=0.6mm},->]
      node[left,pos=0.3] {$\p{\pamove}$} (e)
      (g') edge[decorate,decoration={snake,post length=1mm,amplitude=0.6mm},->]
      node[right,pos=0.3] {$\p$} (f)
      (g) edge[decorate,decoration={snake,post length=1mm,amplitude=0.6mm},->]
      node[right] {$\p'$} (f);
      \path[purple,very thick,shorten > =5pt, shorten < = 5pt,->]
      (f) edge[bend left=40] node[left,pos=0.7] {$\lastf {\p}$} (g')
      (g') edge[bend left=75,looseness=2] node[right,pos=0.9,xshift=0.1cm] {$\firstf {{\pamove}\p'}$} (f);

      \draw[->] (el) -- (fl);
    \end{tikzpicture}}
  \end{minipage}
  \begin{minipage}[b]{0.40\textwidth}
    \scalebox{0.85}{
    \begin{tikzpicture}[scale=0.8,font=\footnotesize,inner sep=2pt,auto,>=stealth]
      \draw (0,0) -- (6,0);
      \draw (0,2) -- (6,2);

      \node (q) at (-0.3,0) {$q$};
      \node (p) at (-0.3,2) {$p$};

      \node[anchor=west] at (-0.3,2.4)
      {$\last {\p{\pamove}} e ~\le~ \last {\p'{\pamove}} e$};
      \node[dot] (p) at (0.5,2) {};
      \node[dot] (p') at ($(p)+(2.5,0)$) {};
      \node[dot] (g') at (4.7,2) {};
      \node[dot,right=0.8cm of g'] (g) {};
      \node[above=0.1cm of g'] {$g'$};
      \node[above right= 0.1cm and 0.2cm of g'] {$\vphantom{g'}{<}$};
      \node[above=0.1cm of g] {$g\vphantom{g'}$};
      \node[dot,label={[name=el]below:{$e$\vphantom{$f$}}}] (e) at (4.3,0) {};
      \node[dot,label={[name=fl]below:{$f$}}] (f) at (5.8,0) {};

      \path
      (p) edge[decorate,decoration={snake,post length=1mm,amplitude=0.6mm},->]
      node[left,xshift=-5pt] {$\p{\pamove}$} (e)
      (p') edge[decorate,decoration={snake,post length=1mm,amplitude=0.6mm},->]
      node[left,pos=0.3] {$\p'{\pamove}$} (e)
      (g') edge[decorate,decoration={snake,post length=1mm,amplitude=0.6mm},->]
      node[right,pos=0.3] {$\p'$} (f)
      (g) edge[decorate,decoration={snake,post length=1mm,amplitude=0.6mm},->]
      node[right] {$\p$} (f);
      \path[purple,very thick,shorten > =5pt, shorten < = 5pt,->]
      (f) edge[bend left=40] node[left,pos=0.8] {$\lastf {\p'}$} (g')
      (g') edge[bend left=75,looseness=2] node[right,pos=0.9,xshift=0.1cm] {$\firstf {{\ppmove}\p}$} (f);

      \draw[->] (el) -- (fl);
    \end{tikzpicture}}
\end{minipage}
\caption{Lemma~\ref{ple-char}, cases 1., 2., and 3.\label{fig:ple-char}}
\end{figure}
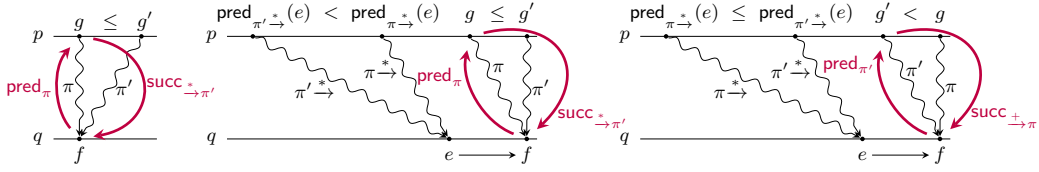

Generally, the relation $\ple_e$ 
can be characterized as follows (cf.\ also Figure~\ref{fig:ple-char}): 

\begin{lemma}\label{ple-char}
  Let $\p,\p' \in \Paths_{p,q}$ with $p,q\in\Procs$, and $f \in E_q$.
  \begin{enumerate}
  \item\label{lem:initle} Assume that there exists no $e$ with $e \prel f$.

    Then, $\p \ple_f \p'$ iff $\last \p f = \botp p$ or
    $\fa \p {\pamove \p'} (f) = f$.
  \item\label{lem:gtle} Assume that there exists $e \in E_q$ such that
    $e \prel f$ and $\p'{\pamove} \pless_e \p{\pamove}$.
    
    Then, $\p \ple_f \p'$ iff $\last {\p} f = \botp p$ or
    $\fa \p {\pamove \p'} (f) = f$.
  \item\label{lem:legt} Assume that there exists $e \in E_q$ such that
    $e \prel f$ and $\p{\pamove} \ple_e \p'{\pamove}$.

    Then, $\p' \pless_f \p$ iff $\last {\p'} f = \botp p$ and
    $\last {\p} f \neq \botp p$, or $\fa {\p'} {\ppmove \p} (f) = f$.
  \end{enumerate}
\end{lemma}

\begin{proof}
  If $\last \p f = \bot$ or $\last {\p'} f = \bot$, the proof of 1., 2.,
  and 3.\ is immediate.
  So we assume this is not the case, and we let $g = \last \p f \in E_p$ and
  $g' = \last {\p'} f \in E_p$.

  We first show that $\p \ple_f \p'$ iff $\fa {\p} {{\pamove}\p'} (f) \le f$.
  Indeed, if $\p \ple_f \p'$, then we have $g \pamove g'$ and
  $(g',f) \in \sem M {\p'}$, hence $(g,f) \in \sem M {{\pamove}\p'}$.
  Then, by definition,
  $\fa \p {{\pamove}\p'} (f) = \first {{\pamove}\p'} g \le f$.
  Conversely, if $\p' \pless_f \p$, i.e., $g' < g$, then by maximality of
  $g' = \last {\p'} {f}$, we have $(g,f) \notin \sem M {{\pamove}\pi'}$,
  hence $f < \first {{\pamove}\pi'} g = \fa \p {{\pamove}\p'} (f)$
  (either $\first {{\pamove}\pi'} g = \top$, or it is an event to the right
  of $f$).

  Similarly, we have $\p' \pless_f \p$ iff $\fa {\p'} {{\ppmove}\p} (f) \le f$.
  So, in all three statements, all that remains to be proved is the equality
  in the left-to-right implications:
    \begin{enumerate}
  \item Assume $f$ is $\prel$-minimal and $\p \ple_f \p'$.
    By the above, we have $\fa \p {{\pamove}\p'} (f) \le f$, and since $f$ is
    $\prel$-minimal, $\fa \p {{\pamove}\p'} (f) = f$.
  \item Assume $e \prel f$, $\p'{\pamove} \pless_e \p{\pamove}$, and
    $\p \ple_f \p'$. In particular, $f' \df \fa \p {{\pamove}\p'} (f)
    = \first{{\pamove}\p'}{g} \le f$.
    Now, suppose $f' < f$ and, therefore, $f' \le e$.
    Notice that $\last{\p'}{f'}\neq\botp{p}$ and $g\leq\last{\p'}{f'}$.  Using
    Lemma~\ref{lem:monotone} (monotonicity), we obtain the following
    contradiction:
    $$
    g\leq\last{\p'}{f'}=\last{\p'{\pamove}}{f'}\leq\last{\p'{\pamove}}{e}
    <\last{\p{\pamove}}{e}\leq\last{\p{\pamove}}{f}=\last{\p}{f}=g \,.
    $$
  \item Assume $e \prel f$, $\p{\pamove} \ple_e \p'{\pamove}$, and
    $\p' \pless_f \p$. In particular, $f'=\fa {\p'} {{\ppmove}\p} (f)\leq f$.
    Now, suppose $f' \le e$.
    Notice that $\last{\p}{f'}\neq\botp{p}$ and $g'<\last{\p}{f'}$.  Using
    Lemma~\ref{lem:monotone} (monotonicity), we obtain the following
    contradiction:
    $$
    g'<\last{\p}{f'}=\last{\p{\pamove}}{f'}\leq\last{\p{\pamove}}{e}
    \leq\last{\p'{\pamove}}{e}\leq\last{\p'{\pamove}}{f}=\last{\p'}{f}=g' \,.
    $$
  \end{enumerate}
  This concludes the proof.
\end{proof}

\section{Constructing the Gossip CFM}\label{sec:preorder-cfm}

In this section, we construct $\Ale$ computing the total preorders $\ple_e$
over a finite set of path expressions $\Pi$. We define the \emph{size} of
$\Pi$ as $\|\Pi\| = \sum_{\pi \in \Pi} |\pi|$, where $|\pi|$ denotes the length
of $\pi$.

\subsection{CFMs for $\fa \p {\p'}$}\label{sec:cfm-labels}

\begin{lemma}\label{lem:last-label}
  Let $\Theta$ be a finite set such that $\bot \notin \Theta$, and
  $\p \in \moves^\ast$ a path expression.
  There exists a \CFM with $|\Theta|^{\mathcal{O}(|\p|)}$ states
  recognizing the set of extended MSCs $(M,\xi)$ with
  $\xi \colon E \to \Theta \times (\Theta \cup \{\bot\})$ such that,
  for all events $e$, $\xi(e)$ is a pair $(\xia e,\xib e)$ such that
  $\xib e = \xia {\last \p e}$, with $\xia \bot = \bot$.
\end{lemma}

\begin{proof}
  Let $\Pi = \{\p' \in \moves^\ast \mid \exists \p'' \in \moves^\ast\text{: }
  \p = \p'\p''\}$ be the set of prefixes of $\p$.
  The state of the \CFM taken at event $e$ will consist of a function
  $\theta(e) \colon \Pi \to \Theta \cup \{\bot\}$ such that, for all $e \in E$ and $\p' \in \Pi$,
  $\tht e {\p'} = \xia {\last {\p'} e}$.
  If $e$ is a send event, the function $\theta(e)$ is sent as a message.
  In order to determine $\tht e {\p_1}$ for all events $e$ and $\p_1 \in \Pi$,
  the \CFM only allows transitions ensuring the following:
  \begin{itemize}
  \item Suppose $\p_1 = \varepsilon$. Then, $\tht e {\p_1} = \xia e$.
  \item Suppose $\p_1 = \p_2\pmove$. If $e$ is $\prel$-minimal, then
    $\tht e {\p_1} = \bot$.
    If $f \prel e$ for some $f$, then $\tht e {\p_1} = \tht f {\p_2}$.

  \item Suppose $\p_1 = \p_2 {\pamove}$.
    If $\tht e {\p_2} \neq \bot$, then $\tht e {\p_1} = \tht e {\p_2}$
    (Lemma~\ref{lem:monotone}).
    If $\tht e {\p_2} = \bot$ and $e$ is $\prel$-minimal, then
    $\tht e {\p_1} = \bot$.
    If $\tht e {\p_2} = \bot$ and $f \prel e$ for some $f$, then
    $\tht e {\p_1} = \tht f {\p_1}$.

  \item Suppose $\p_1 = \p_2 \mmove p q$.  If $e \in E_q$ and there is an event
    $f \in E_p$ such that $f \mrel e$, then $\tht e {\p_1} = \tht f {\p_2}$.
    Otherwise, $\tht e {\p_1} =\bot$.

  \item Suppose $\p_1 = \p_2 \amove a$.  If $\lambda(e) = a$, then
    $\tht e {\p_1} = \tht e {\p_2}$.  Otherwise, $\tht e {\p_1} = \bot$.
  % \qedhere
  \end{itemize}
  Finally, the \CFM checks that, for all events $e$, $\xib e = \tht e \p$,
  i.e., $\xib e = \xia {\last \p e}$.
\end{proof}

We can prove a similar result for $\mathsf{succ}_\p$:

\begin{lemma}\label{lem:succ-label}
  Let $\Theta$ be a finite set such that $\top \notin \Theta$, and
  $\p \in \moves^\ast$ a path expression.
  There exists a \CFM with $|\Theta|^{\mathcal{O}(|\p|)}$ states
  recognizing the set of extended MSCs $(M,\xi)$ with
  $\xi \colon E \to \Theta \times (\Theta \cup \{\top\})$ such that,
  for all events $e$, $\xi(e)$ is a pair $(\xia e,\xib e)$ such that
  $\xib e = \xia {\first \p e}$, with $\xia \top = \top$.
\end{lemma}

\begin{proof}
  Let $\Pi = \{\p'' \in \moves^\ast \mid \exists \p' \in \moves^\ast\text{: }
  \p = \p'\p''\}$ be the set of suffixes of $\p$.
  The state of the \CFM taken at event $e$ will consist of a
  function $\theta(e) \colon \Pi \to
  \Theta \cup \{\top\}$ such that, for all $e \in E$ and $\p' \in \Pi$,
  $\tht e {\p'} = \xia {\first {\p'} e}$.
  If $e$ is a send event, the function $\theta(e)$ is sent as a message.
  In order to determine $\tht e {\p_1}$ for all events $e$ and $\p_1 \in \Pi$,
  the \CFM only allows transitions ensuring the following:
  \begin{itemize}
  \item Suppose $\p_1 = \varepsilon$. Then, $\tht e {\p_1} = \xia e$.
  \item Suppose $\p_1 = \pmove\p_2$. If $e$ is $\prel$-maximal, then
    $\tht e {\p_1} = \top$.
    If $e \prel f$ for some $f$, then $\tht e {\p_1} = \tht f {\p_2}$.

  \item Suppose $\p_1 = {\pamove} \p_2$.
    If $\tht e {\p_2} \neq \top$, then $\tht e {\p_1} = \tht e {\p_2}$
    (Lemma~\ref{lem:monotone}).
    If $\tht e {\p_2} = \top$ and $e$ is $\prel$-maximal, then
    $\tht e {\p_1} = \top$.
    If $\tht e {\p_2} = \top$ and $e \prel f$ for some $f$, then
    $\tht e {\p_1} = \tht f {\p_1}$.

  \item Suppose $\p_1 = \mmove p q \p_2$.  If $e \in E_p$ and there is an event
    $f \in E_q$ such that $e \mrel f$, then $\tht e {\p_1} = \tht f {\p_2}$.
    Otherwise, $\tht e {\p_1} =\top$.

  \item Suppose $\p_1 = \amove a \p_2$.  If $\lambda(e) = a$, then
    $\tht e {\p_1} = \tht e {\p_2}$.  Otherwise, $\tht e {\p_1} = \top$.
  % \qedhere
  \end{itemize}
  Finally, the \CFM checks that, for all events $e$, $\xib e = \tht e \p$,
  i.e., $\xib e = \xia {\first \p e}$.
\end{proof}

As a corollary, we obtain a \CFM for $\fa \p {\p'}$:

\begin{lemma}
  \label{lem:fa-label}
  Let $\Theta$ be a finite set such that $\Theta \cap \{\bot,\top\} =
  \emptyset$, $p,q \in \Procs$, and $\p,\p' \in \Paths_{p,q}$.
  There exists a \CFM
  with $|\Theta|^{\mathcal{O}(|\pi|+|\pi'|)}$ states
  recognizing the set of extended MSCs $(M,\xi)$ with 
  $\xi\colon E \to \Theta \times (\Theta \cup \{\bot,\top\})$ such that, for all
  events~$e \in E_q$, $\xi(e)$ is a pair $(\xia e, \xib e)$ such that
  $\xib e = \xia {\fa \p {\p'} (e)}$.
\end{lemma}

We are now ready to prove that there exists a \CFM $\Afa \p {\p'}$
that determines, for each event $e$, whether $\fa \p {\p'} (e) = e$.

\begin{lemma}
  \label{lem:Afa}
  Let $\p,\p' \in \Paths_{p,q}$ with $p,q\in\Procs$.
  There exists a \CFM $\Afa \p {\p'}$ over $P$ and $\Sigma\times\{0,1\}$ 
  with $2^{\mathcal{O}(|\pi|+|\pi'|)}$ states
  that recognizes the set of MSCs $(M,\gamma)$ such that, for all events $e$
  on process $q$, we have $\gamma(e) = 1$ iff $\fa \p {\p'} (e) = e$.
\end{lemma}

\begin{proof}
We denote by $L$ the set of MSCs $(M,\gamma)$ such that, for all events
$e$ on process $q$, $\gamma(e) = 1$ iff $\fa \p {\p'} (e) = e$.
To ensure that the input MSC is in $L$, the \CFM
$\Afa \p {\p'}$ will use a coloring of the events of process $q$, constructed
in such a way that, for all events $e$ on process~$q$, the events $e$ and
$\fa {\p} {\p'} (e)$ have the same color iff they are equal.

Formally, we consider doubly extended MSCs $(M,\gamma,\zeta)$ with $\gamma
\colon E \to \{0,1\}$ and $\zeta \colon E \to \{\colone,\coltwo,\zcolone,
\zcoltwo\}$. As usual, we define $\zeta(\bot) = \bot$ and $\zeta(\top) = \top$.
Let $\tilde{L}$ be the set of MSCs $(M,\gamma,\zeta)$ such that the following hold:
\begin{enumerate}
  \item Denoting by $e_1 < e_2 < \cdots < e_k$ the events on process $q$
    with $\gamma(e_i) = 1$, we have $\zeta(e_i) = \colone$ if $i$ is
    odd, $\zeta(e_i) = \coltwo$ if $i$ is even, and $\zeta(e) \in
    \{\zcolone,\zcoltwo\}$ if $e \in E_q\setminus \{e_1, \ldots, e_k\}$.

    Intuitively, $\zeta(e)$ will be a color computed (if $\gamma(e) = 1$)
    or guessed (if $\gamma(e) = 0$) by $\Afa \p {\p'}$.

  \item For all $e \in E_q$, $\gamma(e) = 1$ iff
    $\zeta(e) = \zeta(\fa \p {\p'} (e))$.
\end{enumerate}
We first show that there exists a \CFM accepting $\tilde{L}$.  First, applying
Lemma~\ref{lem:fa-label} with $\Theta = \{\colone,\coltwo,\zcolone,\zcoltwo\}$,
we know that there exists a \CFM accepting the set of MSCs $(M,\gamma,\xi)$
with $\xi \colon E \to \Theta \times (\Theta \cup \{\bot,\top\})$ such that,
for all events $e$, $\xi(e) = (\xia e, \xib e) =
(\xia e, \xia {\fa \p {\p'} (e)}$.
We then restrict the transitions of this \CFM so that it additionally
checks that, for all events $e$ on process $q$, $\gamma(e) = 1$ iff
$\xia e = \xib e$.
By projection onto the first component of $\xi$, we obtain a \CFM 
accepting $\tilde{L}$.

We define $\Afa \p {\p'}$ as the \CFM recognizing the projection of $\tilde{L}$
on $\Sigma \times \{0,1\}$. We claim that $L(\Afa \p {\p'}) = L$.

We first prove the left-to-right inclusion.
Suppose $(M,\gamma,\zeta) \in \tilde{L}$, with $e_1,\ldots,e_k$ defined 
as above. Towards a contradiction, assume $(M,\gamma) \not\in L$.
For all events $e\in E_q\setminus\{e_1,\ldots,e_k\}$, we have
$\zeta(e)\neq\zeta(\fa \p {\p'} (e))$, hence $\fa \p {\p'} (e) \neq e$.
So there exists $g_0 \in \{e_1,\ldots,e_k\}$ such that $g_0 \neq \fa \p {\p'} (g_0)$.
For all $i \in \mathbb{N}$, let $g_{i+1} = \fa \p {\p'} (g_i)$. 
Note that $g_i\in\{e_1,\ldots,e_k\}$ implies that $\fa \p {\p'} (g_i)\in E_q$
and $\zeta(g_{i+1})=\zeta(g_{i})\in\{\colone,\coltwo\}$, hence 
$g_{i+1}\in\{e_1,\ldots,e_k\}$.  Suppose $g_0 < g_1$
(the case $g_1 < g_0$ is similar).  Take $g_0 < h_0 < g_1$ such that $\zeta(h_0)
\in \{\colone,\coltwo\}$ and $\zeta(h_0) \neq \zeta(g_0)$.
Again, for all $i \in \mathbb{N}$, let $h_{i+1} = \fa \p {\p'} (h_i)$.
Note that all $g_0,g_1,\ldots$ have the same color, and all $h_0,h_1,\ldots$
carry the complementary color.  Thus, $g_i \neq h_i$ for all $i \in \mathbb{N}$.
But, by Lemma~\ref{lem:monotone2}, this implies $g_0 < h_0 < g_1 < h_1 < \ldots$
which contradicts the fact that we deal with finite MSCs.

Next, we show that $L \subseteq L(\Afa \p {\p'})$.
Suppose $(M,\gamma) \in L$. Let $E_0 = \{e \in E_q \mid \gamma(e) = 0\}
= \{e \in E_q \mid \fa \p {\p'} (e) \neq e\}$ and $E_1 = \{e \in E_q \mid
\gamma(e) = 1\} = \{e \in E_q \mid \fa \p {\p'} (e) = e\}$.
Consider the graph $G = (E_q,\{(e,\fa \p {\p'} (e)) \mid e \in E_q 
\wedge \fa \p {\p'} (e)\in E_q\})$.
Every vertex has outdegree at most 1, and, since $\fa \p {\p'}$ is
monotone, there are no cycles except for self-loops.
So the restriction of $G$ to $E_0$ is a forest,
and there exists a $2$-coloring $\chi\colon E_0 \to \{\zcolone,\zcoltwo\}$ such
that, for all $e\in E_0$ with $\fa \p {\p'} (e)\in E_0$, we have
$\fa \p {\p'} (e) \in \{\bot,\top\}$ or $\chi(e) \neq \chi(\fa \p {\p'} (e))$.
Define $\zeta \colon E \to \{\colone,\coltwo,\zcolone,\zcoltwo\}$ by
$\zeta(e) = \chi(e)$ for $e \in E_0$ and as in Condition 1.\
for $e \in E_1$. Notice that Condition 2.\ is satisfied. Hence,
$(M,\gamma,\zeta) \in \tilde{L}$ and $(M,\gamma) \in L(\Afa \p {\p'})$.
\end{proof}

\subsection{The Gossip CFM}\label{sec:gossip-cfm}

Let $p,q \in \Procs$ and $\Paths$ be a finite subset of $\Paths_{p,q}$.
We are now in a positon to build a (non-deterministic) \CFM that outputs,
at every event $e \in E_q$,
the restriction of ${\ple_e}$ to $\Paths \times \Paths$.

\begin{lemma}
  \label{lem:ord}
  Let $\Ord$ be the set of preorders over $\Paths$.
  There exists a \CFM $\Ale$ over $P$ and $\Sigma \times \Ord$
  with $2^{\mathcal{O}(\|\Pi\|^2)}$ states
  that recognizes the set of MSCs $(M,\gamma)$ such that $\gamma(e) = {\ple_e}$.
\end{lemma}

\begin{proof}
  Without loss of generality, we can assume that, for all $\p \in \Paths$,
  we have $\p{\pamove} \in \Paths$ or $\p = \p'{\pamove}$ for some
  $\p' \in \Gamma^\ast$.
  In addition, we will identify path expressions $\p{\pamove}{\pamove}$
  and $\p{\pamove}$, observing that we have
  $\sem M {\p{\pamove}{\pamove}} = \sem M {\p{\pamove}}$.
  With this convention, we can always assume that, if $\p \in \Paths$, then
  $\p{\pamove} \in \Paths$, while keeping $\Paths$ finite (and of linear size).
  
  By Lemma~\ref{lem:Afa} (and since $\Paths$ is finite), $\Ale$ can determine,
  for each event $e$ and all path expressions $\p,\p' \in \Paths \cup
  \{{\pamove}\p \mid \p \in \Paths\} \cup \{{\ppmove}{\p} \mid \p \in \Paths\}$,
  whether $\fa \p {\p'} (e) = e$.
  The \CFM then checks that, for all $f \in E_q$ and $\p,\p' \in \Paths$,
  $(\p,\p') \in \gamma(f)$ iff one of the following holds
  (cf.\ Lemma~\ref{ple-char}):
  \begin{itemize}
  \item $f$ is minimal on process $q$, and $\last \p f = \botp p$
    or $\fa \p {{\pamove}\p'} (f) = f$.
  \item $e \prel f$, $(\p{\pamove},\p'{\pamove}) \notin \gamma(e)$,
    and $\last \p f = \bot$ or $\fa \p {{\pamove}\p'} (f) = f$.
  \item $e \prel f$, $(\p{\pamove},\p'{\pamove}) \in \gamma(e)$,
    $\fa {\p'} {{\ppmove}\p} (f) \neq f$,
    and $\last \p f = \botp p$ or $\last {\p'} f \neq \botp p$.
    \qedhere
  \end{itemize}
\end{proof}

In fact, for the gossip problem, one needs only a particular set of path
expressions.  For a sequence $w = p_1 \ldots p_n \in \Procs^+$ of pairwise
distinct processes, we define the path expression $\pexpr{w}$ by
$\pexpr{w}={\ppmove}$ if $n=1$, and
$\pexpr{w}={\pamove}{\mmove{p_1}{p_2}}{\pamove}{\mmove{p_2}{p_3}} \ldots
{\pamove}{\mmove{p_{n-1}}{p_n}}{\pamove}$ if $n\geq2$.  Let $\SimplePaths$ be
the set of all those path expressions (which is finite).  Finally, given
processes $p,q \in \Procs$, we define $\SimplePaths_{p,q} =
\Paths_{p,q}\cap\SimplePaths$.
We have ${<} = \bigcup_{\p \in \SimplePaths} \sem M {\p}$.
Moreover, for all $e \in E_q$, $\lastp p e = \max \{ \last \p e \mid \p \in
\SimplePaths_{p,q} \}$.

We can now apply Lemma~\ref{lem:ord} to all sets $\SimplePaths_{p,q}$ to obtain the desired gossip \CFM $\Agossip$:

\begin{theorem}
  \label{gossipcfm}
  There exists a CFM $\Agossip$ with $|\Sigma|^{2^{\mathcal{O}(|P| \log |P|)}}$ states
  that recognizes the set of extended MSCs
  $((E,\prel,\mrel,\ploc,\lambda),\xi)$ such that, for all events $e \in E$,
  $\xi(e)$ is the function from $\Procs$ to $\Sigma \cup \{\bot\}$ defined by
  $\xi(e)(p) = \lambda(\lastp p e)$. 
\end{theorem}

\begin{proof}
  The \CFM $\Agossip$ guesses, for all $e \in E_q$, some $\p \in \SimplePaths_{p,q}$.
  Using Lemma~\ref{lem:ord}, it verifies $\lastp p e = \last \p e$.
  Moreover, using Lemma~\ref{lem:last-label}, it checks that $\xi(e) = \lambda(\last \p e)$.
\end{proof}

Next, we show that $\Agossip$ is, unavoidably, non-deterministic.
Following \cite{HenriksenJournal,GKM07,Kuske01}, we call a \CFM $\C = ((\A_p)_{p \in
\Procs},\Msg,\Acc)$ \emph{deterministic} if, for all processes $p$ and
transitions $t_1=(s_1,\act_1,s_1')$ and $t_2=(s_2,\act_2,s_2')$ of $\A_p$ such that $s_1 = s_2$ and
$\tlabel(t_1) = \tlabel(t_2)$, the following hold:
\begin{itemize}
\item If $t_1$ and $t_2$ are internal transitions, then $s_1' = s_2'$.
\item If $t_1$ and $t_2$ are send transitions such that $\receiver(t_1) = \receiver(t_2)$, then $s_1' = s_2'$ and $\tmsg(t_1) = \tmsg(t_2)$.
\item If $t_1$ and $t_2$ are receive transitions such that $\sender(t_1) = \sender(t_2)$ and $\tmsg(t_1) = \tmsg(t_2)$, then $s_1' = s_2'$.
\end{itemize}

\begin{proposition}
  There is no deterministic gossip \CFM for $|\Sigma| \ge 2$ and $|P| \ge 3$.
\end{proposition}

\begin{proof}
  Let $P = \{p,q,r\}$ and $\Sigma = \{\bbullet,\abullet,\nbullet\}$.
  The symbol $\nbullet$ will only be used for clarity, and could be replaced
  arbitrarily with $\bbullet$ or $\abullet$.
  We show that there exists no deterministic \CFM recognizing the set $L$ of
  MSCs $M = (E,\prel,\mrel,\ploc,\lambda)$ such that for all $e \in E_q$,
  $\lambda(e) = \lambda(\lastp p e)$.
  As a consequence, there is no deterministic gossip \CFM over $P$ and $\Sigma$.
    
  Assume that there exists a deterministic \CFM
  $\A = (\A_p,\A_q,\A_r,\Msg,\Acc)$ such that $L(\A) = L$.
  Fix $n > |S_q|^2$, where $S_q$ is the set of states of $\A_q$.
  For all $k \in \{0,\ldots,n-1\}$, we define
  an MSC $M^k = (E,\prel,\mrel^k,\ploc,\lambda^k)$,
  as depicted in Figure~\ref{fig:det-counter-example}
  (where $n = 5$ and $k = 2$):
  \begin{itemize}
  \item $E_p = \{e_i \mid 0 \le i < 2n\}$,
    $E_q = \{f_i \mid 0 \le i < 2n\}$,
    and $E_r = \{g_i \mid 0 \le i < 2n\}$,
    with $e_0 \prel e_1 \prel \cdots \prel e_{2n-1}$,
    $f_0 \prel f_1 \prel \cdots \prel f_{2n-1}$,
    and $g_0 \prel g_1 \prel \cdots \prel g_{2n-1}$.
  \item For all $0 \le i < k$, $e_{2i} \mrel^k f_{i}$, and for all
    $k \le i < n$, $e_{2i} \mrel^k f_{n + i}$.

    For all $0 \le i < n$, $e_{2i+1} \mrel^k g_{2i}$, and $g_{2i+1} \mrel^k f_{k+i}$.
  \item For all $0 \le i < n$,
    $\lambda^k(e_{2i}) = \bbullet$ and $\lambda^k(e_{2i+1}) = \abullet$.
    
    For all $f \in E_q$, $\lambda^k(f) = \lambda^k(\lastp p e)$.
    That is, for all $0 \le i < 2k-1$, $\lambda^k(f_i) = \bbullet$,
    and for all $2k - 1 \le i < n$, $\lambda^k(f_{i}) = \abullet$.
    
    For all $g \in E_r$, $\lambda^k(g) = \nbullet$.
  \end{itemize}
  Clearly, $M^k \in L(\A)$.
  Let $s_k$ and $t_k$ be the states associated respectively with $f_{k-1}$
  (or the initial state of $\A_q$ if $k = 0$)
  and $f_{k+n-1}$ in the unique run $\rho^k$ of $\A$ on $M^k$.
  That is, if $k > 0$, $s_k = \target(\rho^k(f_{k-1}))$ and
  $t_k = \target(\rho^k(f_{k+n-1}))$.
  
  \begin{figure}[h]
  \centering
    \begin{tikzpicture}[semithick,>=stealth,xscale=1]
      \draw (-0.5,0) -- (12,0);
      \draw (-0.5,1.5) -- (12,1.5);
      \draw (-0.5,3) -- (12,3);
      
      \foreach \i
      [evaluate=\i as \x using int(2*\i), evaluate=\i as \y using int(2*\i+1)]
      in {0,...,4} {
        \node[bcirc,label=above:$e_{\x}$] (e\x) at (1.5*\i,3) {};
        \node[acirc,label=above:$e_{\y}$] (e\y) at (1.5*\i+0.6,3) {};
        \node[ncirc,label=below:$g_{\x}$] (g\x) at (1.5*\i+0.6,1.5) {};
        \node[ncirc,label=below:$g_{\y}$] (g\y) at (1.5*\i+1.2,1.5) {};
        \draw[->] (e\y) -- (g\x);
      }
      \foreach \i in {0,1} {
        \node[bcirc,label=below:$f_\i$] (f\i) at (1.5*\i,0) {};
      }
      \foreach \i in {2} {
        \node[bcirc,label=below:$f_\i$] (f\i) at (1.3*\i,0) {};
      }
      \foreach \i in {3,...,9} {
        \node[acirc,label=below:$f_\i$] (f\i) at (1.3*\i,0) {};
      }
      \foreach \i [evaluate=\i as \x using int(2*\i)] in {0,1} {
        \draw[->] (e\x) -- (f\i);
      }
      \foreach \i [evaluate=\i as \x using int(2*\i),
      evaluate=\i as \y using int(5+\i)] in {2,3,4} {
        \draw[->] (e\x) -- (f\y);
      }
      \foreach \i [evaluate=\i as \x using int(2*\i+1),
      evaluate=\i as \y using int(2+\i)] in {0,...,4} {
        \draw[->] (g\x) -- (f\y);
      }

      \node at (-1,0) {$q$};
      \node at (-1,1.5) {$r$};
      \node at (-1,3) {$p$};

      \draw[decorate,decoration={brace,amplitude=5pt,raise=0.7cm}]
      (f1) -- (f0) node[midway,yshift=-1cm] {$k$};
      \draw[decorate,decoration={brace,amplitude=5pt,raise=0.7cm}]
      (f6) -- (f2) node[midway,yshift=-1cm] {$n$};
      \draw[decorate,decoration={brace,amplitude=5pt,raise=0.7cm}]
      (f9) -- (f7) node[midway,yshift=-1cm] {$n-k$};

      \node[above left = 0cm of f1] {$s_k$};
      \node[above left = 0cm of f6] {$t_k$};
    \end{tikzpicture}
    \caption{Definition of $M^k$\label{fig:det-counter-example}}
  \end{figure}
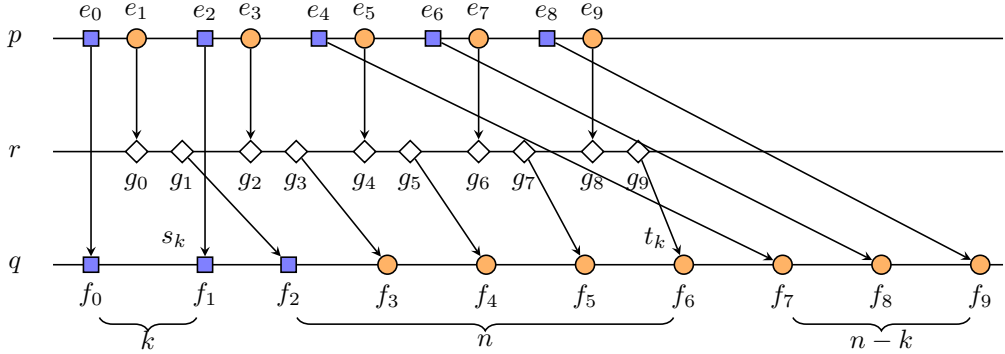

  Note that for all $k$, the sequence of send and receive actions performed by
  process $p$ or process $r$ in $M^k$ are the same, so the runs of
  $\A$ on MSCs $M^k$ only differ on process $q$.
  In particular, the sequence of $n$ messages sent by process $r$ to process $q$
  is the same for all $k$.
  Moreover,
  since $n > |S_q|^2$, there exist $0 \le k < k' < n$ such that $s_k = s_{k'}$
  and $t_k = t_{k'}$.
  We can then combine the runs of $\A$ on $M^k$ and $M^{k'}$
  to define a run where process $q$ receives the messages from process $p$ and
  $r$ in the same order as in $M^k$, but behaves as in $M^{k'}$ in the middle
  part where it receives the $n$ messages from process~$r$. More precisely,
  let $M = (E,\prel,\mrel^k,\ploc,\lambda)$, where $(E,\prel,\mrel^k,\ploc)$ is
  as in $M^k$, and $\lambda$ is defined as follows: for all $0 \le i < k+k'-1$,
  $\lambda(f_{i}) = \bbullet$, and for all $k+k'-1 \le i < n$,
  $\lambda(f_{i}) = \abullet$.
  Then $M \in L(\A)$, but $M \notin L$.
\end{proof}

% ------------------------------------------------------------
% ------------------------------------------------------------
% ------------------------------------------------------------

\section{Linear-Time Temporal Logic}
\label{sec:logic}

The transformation of temporal-logic formulas into automata has many applications, ranging from synthesis to verification. Temporal logics are well understood in the realm of sequential systems where formulas can reason about linearly ordered sequences of events. As we have seen, executions of concurrent systems are actually partially ordered. Over partial orders, however, there is no longer a canonical temporal logic like LTL over words. There have been several attempts to define natural counterparts over Mazurkiewicz traces (see \cite{GK-fi07} for an overview). All of them are less expressive than asynchronous automata \cite{Zielonka87}, a standard model of shared-memory systems. We will show below that this is still true when formulas are interpreted over MSCs and the system model is given in terms of \CFMs.

\smallskip

Many temporal logics over partial orders are captured by the following generic language, which we call \TL.
The set of \TL formulas is defined as follows:
\[
  \varphi ::= a \mid p \mid \varphi \lor \varphi \mid \lnot \varphi \mid
  \Co \varphi \mid \varphi \Until \varphi \mid \varphi \Since \varphi
  \qquad \text{where } a \in \Sigma\text{, } p \in P \, .
\]
A formula $\varphi \in \TL$ is interpreted over events of MSCs.
We say that $M,e \models a$ if $\lambda(e) = a$;
similarly, $M,e \models p$ if $\ploc(e) = p$.
The $\Co$ modality jumps to a parallel event:
$M,e \models \Co \varphi$ if there exists $f \in E$ such that $e \not\le f$,
$f \not\le e$, and $M,f \models \varphi$.
We use strict versions of until and since:
\[
  \begin{array}{lcl}
    M,e \models \varphi_1 \Until \varphi_2
    & \quad\text{if}\quad & \text{there exists $f \in E$ such that }
                            e < f \text{ and } M,f \models \varphi_2 \\
    & & \text{and, for all } e < g < f,\ M,g \models \varphi_1 \, \\
    M,e \models \varphi_1 \Since \varphi_2
    & \quad\text{if}\quad & \text{there exists $f \in E$ such that }
                            f < e \text{ and } M,f \models \varphi_2 \\
    & & \text{and, for all } f < g < e,\ M,g \models \varphi_1 \, .
  \end{array}
\]

This temporal logic and others have been studied in the context of Mazurkiewicz
traces \cite{GK-fi07,Thiagarajan94,DiekertG06}.  The logic introduced by Thiagarajan in
\cite{Thiagarajan94} uses an until modality $\Up p$ corresponding to the usual LTL
(non-strict) until for
a single process $p$, together with a unary modality $\Op p$ interpreted as
follows: $\Op p \varphi$ holds at $e$ if the first event on process $p$ that is
not in the past of $e$ satisfies $\varphi$.  Other interesting modalities are
$\X{p}$ and $\Y{p}$ with the following meaning: $\X{p}$ moves to the first event on process $p$ in the
strict future of the current event, while $\Y{p}$ moves to the last event on
process $p$ that is in the strict past of the current event.  All these
modalities can be expressed in \TL:
\begin{align*}
  \X{p}\varphi & \df \neg p \Until(p\wedge\varphi) & \varphi_1 \Up p \varphi_2 & \df (p \land \varphi_2) \lor\Bigl((\lnot p \lor \varphi_1) \land \Bigl((\lnot p \lor \varphi_1) \Until (p \land \varphi_2)\Bigl)\Bigr)
  \\
  \Y{p}\varphi & \df \neg p \Since(p\wedge\varphi)
  &   \Op p \varphi & \df \Y p \X p \varphi
  \lor \Co \bigl(p \land \lnot \Y p \mathit{true} \land \varphi \bigr)
  \lor \X p \bigl(\lnot \Y p \mathit{true} \land \varphi \bigr)
\end{align*}

It turns out that we can exploit our gossip protocol to translate every \TL formula into an equivalent \CFM:

\begin{theorem}
  For all $\varphi \in \TL$, there exists a \CFM $\A_\varphi$ over $\Procs$ and $\Sigma \times \{0,1\}$
  with $2^{|\varphi|^{\mathcal{O}(|P| \log |P|)}}$ states recognizing
  the set of MSCs $(M,\gamma)$ such that, for all events $e$, $\gamma(e) = 1$
  iff $M,e \models \varphi$.
\end{theorem}

\begin{proof}
  We construct $\A_\varphi$ by induction on $\varphi$.
  The cases $\varphi = a$, $\varphi = p$, $\varphi = \lnot \psi$, and
  $\varphi_1\vee\varphi_2$ are straightforward.
  For $\varphi = \Co \psi$, we compose $\A_\psi$ with a \CFM that
  tests, for each event~$e$, whether it is parallel to some $1$-labeled event.
  The existence of such a CFM (with $2^{2^{\mathcal{O}(|P| \log |P|)}}$ states)
  has been shown in \cite[Lemma 14]{BFG-stacs18}.

  Suppose that we have \CFMs $\A_{\varphi_1}$ and $\A_{\varphi_2}$ for
  $\varphi_1$ and $\varphi_2$.  The input MSCs of $\A_{\varphi_1 \Since
  \varphi_2}$ will be ``pre-labeled'' using $\A_{\varphi_1}$ and
  $\A_{\varphi_2}$, and by projection we can assume that we work with MSCs over
  an alphabet $\{a,b,c,d\}$ where $a$ stands for $\varphi_1 \land \varphi_2$,
  $b$ stands for $\varphi_1 \land \lnot \varphi_2$, $c$ stands for $\lnot
  \varphi_1 \land \varphi_2$, and $d$ stands for $\lnot \varphi_1 \land \lnot
  \varphi_2$.  So the construction of $\A_{\varphi_1 \Since \varphi_2}$ comes
  down to the construction of a \CFM over $\{a,b,c,d\}$ for the formula $(a \lor
  b) \Since (a \lor c) \equiv \bigvee_{p,q \in P} \varphi_{p,q}$ where
  $\varphi_{p,q}= q \land \Big((a \lor b) \Since (p \land (a \lor c))\Big)$.
  Moreover, since ${<}=\bigcup_{\p\in\SimplePaths}\sem{\p}{M}$, it is not
  difficult to check that, for all $e \in E_p$, we have: $M,e \models
  \varphi_{p,q}$ iff
  \begin{align*}
    & \max\left\{\last \p e \mid \p \in a\cdot\SimplePaths_{p,q} \cup
      c\cdot\SimplePaths_{p,q} \right\} \\
    > {}
    & \max\left\{\last \p e \mid \p \in \textstyle \bigcup_{r \in P}
      \SimplePaths_{p,r} \cdot c \cdot \SimplePaths_{r,q}
      \cup \SimplePaths_{p,r} \cdot d \cdot \SimplePaths_{r,q} \right\} \, .
  \end{align*}
  Indeed, this can be read as ``the last event $f \in E_q$ satisfying
  $a \lor c$ in the past of $e$ happens after the last event $g \in E_q$
  such that there exists $h$ with $g<h<e$ which is not labeled $a$ or~$b$''.
  Moreover, by Lemma~\ref{lem:ord}, this property can be tested by
  a \CFM.
  
  As CFMs are closed under mirror languages, we can also construct a
  CFM for $\varphi_1\Until\varphi_2$.
\end{proof}

Note that this result is orthogonal to all other known translations of logic
formulas into \emph{unbounded} \CFMs \cite{BolligJournal,BFG-stacs18,BKM-lmcs10}.

% ------------------------------------------------------------
% ------------------------------------------------------------
% ------------------------------------------------------------

\section{Conclusion}\label{sec:conclusion}

We studied the gossip problem in a message-passing environment with unbounded FIFO channels.
Our non-deterministic protocol is of own interest but also sheds light on the expressive power of communicating finite-state machines.
It allows us to embed well-known temporal logics into CFMs, i.e., properties that typically use three first-order variables.
We believe that we can go further and exploit gossiping to capture even more expressive logics and other high-level specifications
based on the notion of message sequence graphs. We leave this to future work.

% ------------------------------------------------------------
% ------------------------------------------------------------
% ------------------------------------------------------------

% References

\bibliography{lit}

% ------------------------------------------------------------
% ------------------------------------------------------------
% ------------------------------------------------------------
\end{document}